\documentclass[11pt, a4paper]{article}
\usepackage[a4paper,
    left=25mm,
    right=25mm,
    top=15mm]{geometry}
\usepackage[utf8]{inputenc}
\usepackage{amsmath,amsthm,amssymb,amsfonts,color}
\newtheorem{theorem}{Theorem}[section]

\newtheorem{proposition}[theorem]{Proposition}

\usepackage[colorlinks=true,
	    linkcolor=black,
            citecolor=black,      
            urlcolor=black]{hyperref}
\usepackage{cleveref}
\usepackage{bm}
\newcommand{\cu}[1]{
	\ifcat\noexpand#1\relax
	\bm{#1}
	\else
	\mathbf{#1}
	\fi
}

\usepackage{amsmath,amsthm,amssymb,amsfonts,color,bm}
\usepackage[round, authoryear]{natbib}
\usepackage{graphicx}
\usepackage{caption, subcaption}
\usepackage{xcolor}
\usepackage{listings}
\usepackage{float}
\usepackage{graphicx}
\usepackage{latexsym}
\usepackage{titlesec}
\titleformat{\abstract}[runin]
	{\normalfont\bfseries}
	{\thesection.}{.2em}{}[.]
\titleformat{\section}[runin]
	{\normalfont\bfseries}
	{\thesection.}{.2em}{}[.]

\titleformat{\subsection}[runin]
	{\normalfont\bfseries}
	{\thesubsection.}{.2em}{}[.]

\newcommand{\R}{\mathbb{R}}
\newcommand{\E}{\mathbb{E}}
\newcommand{\PP}{\mathbb{P}}
\newcommand{\dirp}{\mathrm{DP}}

\newcommand{\dd}{\mathrm{d}}
\newcommand{\dnormal}{\mathrm{N}}
\newcommand{\dbeta}{\mathrm{Be}}
\newcommand{\iid}{\stackrel{\mathrm{i.i.d}}{\sim}}
\newcommand{\ind}{\stackrel{\mathrm{ind}}{\sim}}
\newcommand{\dotind}{\stackrel{\mathrm{ind}}{{.}\sim}}

\newcommand{\indic}{\mathbf{1}}

\usepackage{authblk}
\title{\bf{System identification using Bayesian neural networks with nonparametric noise models}}
\author[1]{Christos Merkatas\thanks{Corresponding author: \href{christos.merkatas@aalto.fi}{christos.merkatas@aalto.fi}}}
\author[1]{Simo S\"arkk\"a}
\affil[1]{Department of Electrical engineering and Automation, Aalto University, Espoo 02150, Finland}
\vspace{0.5cm}
\begin{document}

\maketitle

\begin{abstract}    
System identification is of special interest in science and engineering. This article is concerned with a system identification problem arising in stochastic dynamic systems, where the aim is to estimate the parameters of a system along with its unknown noise processes. In particular, we propose a Bayesian nonparametric approach for system identification in discrete time nonlinear random dynamical systems assuming only the order of the Markov process is known. The proposed method replaces the assumption of Gaussian distributed error components with a highly flexible family of probability density functions based on Bayesian nonparametric priors. Additionally, the functional form of the system is estimated by leveraging Bayesian neural networks which also leads to flexible uncertainty quantification. Asymptotically on the number of hidden neurons, the proposed model converges to full nonparametric Bayesian regression model. A Gibbs sampler for posterior inference is proposed and its effectiveness is illustrated on simulated and real time series.  \newline
\\
{\noindent \bf Keywords:} Bayesian nonparametrics; System identification; Infinite mixture models; Nonlinear time series; Geometric stick breaking
\end{abstract}

\section{Introduction}
\label{intro}
Random dynamical systems, that is, dynamical systems that evolve through the presence of noise, play an important role in science and engineering as they can be used to model and explain dependencies in physical phenomena changing in time. This noise may be observational, attributed to measurement errors from the measuring device, or dynamical which compensates for modeling error or both. Common areas of application include but are not limited to neuroscience, finance \citep{franses2000non, small2005applied, ozaki2012time}, and signal processing \citep{tse2003chaos}. 

In this article, we consider dynamical systems that are observed in the form of a time series and the mathematical model describing the process is given by
\begin{equation}\label{eq:general_model}
    y_t = g(y_t^{\rho};\theta) + z_t,\quad t=1,2,\ldots
\end{equation}
where $y_t$ is the observed process and the vector $y_t^{\rho}:=(y_{t-1},\ldots, y_{t-\rho})$ consists of a set of lagged versions of the process $y_t$ which allows the system to be written in the form of a regression model. Additionally, $\theta$ is a vector of parameters $\theta:=(\theta_1,\ldots, \theta_m)$  and $\{z_t\}_{t\geq 1}$ is a collection of i.i.d. zero-mean noise components with finite variance. For simplicity, we assume that we have at our disposal a time series of length $n$ generated from \eqref{eq:general_model} and that the vector of $\rho$ initial states is known.
    
System identification refers to the problem of accurate estimation of $\hat{g}$ or, equivalently, $\hat{\theta}$ that describes the observed process well enough in order to understand the dynamics and use it for control and/or prediction of future unobserved values of the system \citep{sjoberg1995nonlinear, ljung1999system, nelles2001nonlinear,keesman2011system}. System identification is typically done with some statistical method for estimating the parameters based on the available data. 
  
For a discrete-time scalar time series $\{y_t\}_{t\geq 1},$ the Box--Jenkins \citep{box2015time} modeling approach is to assume an autoregressive integrated moving average $\mathrm{ARIMA}(p,d,q)$ model for $\{y_t\}.$ In this framework, a non-stationary time series is differentiated $d$ times to become stationary and then a suitable model for representing the process $y_t$ is $\mathrm{\Omega}_p(B)y_t = \mathrm{\Psi}_q(B)z_t,$ where 
    \begin{align}
    	\mathrm{\Omega}_p(B) &= 1 - \omega_1 B - \cdots -\omega_p B^p,\nonumber\\
	\mathrm{\Psi}_q(B) &= 1 - \psi_1 B - \cdots -\psi_q B^q
    \end{align}
  are polynomials of degree $p,q$ respectively of the backshift operator $B,$ with $B^\rho y_t = y_{t-\rho}.$ The polynomial degrees, and the coefficients $\theta=(\omega_1,\ldots,\omega_p, \psi_1,\ldots,\psi_q)$ of the polynomials are estimated from the observed data. 
    These models work extremely well for linear Gaussian time series, however, they are not sufficient for nonlinear or chaotic time series modeling. To this end, more sophisticated, nonlinear models for the function $g$ have been developed. 
 Some of the popular choices include, among others, the bilinear model, the threshold autoregressive model along with the extension to self-exciting threshold autoregressive model \citep{tong1980threshold} and the autoregressive conditional heteroscedasticity (ARCH) model \citep{engle1982autoregressive}. Moreover, black-box nonlinear models using artificial neural networks have been successfully used for nonlinear time series estimation and prediction
\citep{chen1992neural, sjoberg1994neural, sjoberg1995nonlinear, zhang2003time, aladag2009forecasting}.
  
    The so-called Bayesian approaches \citep{gelman2013bayesian} incorporate prior information in terms of probability distributions to the estimation process. Bayesian methods
    have been used to estimate the parameters of maps responsible for the generation of chaotic signals contaminated with Gaussian noise \citep{luengo2002bayesian, meyer2000bayesian, pantaleon2000bayesian, pantaleon2003bayesian}. In all cases above, the general functional form of the system in question is known and the aim is to estimate the control parameters. Bayesian methods have also been used \citep{pantaleon2003bayesian} to estimate piecewise Markov systems. Slightly similarly, in the context of Bayesian estimation but, regarding dynamically noise-corrupted signals and taking a black-box approach, \citep{matsumoto2001reconstructions, nakada2005bayesian} developed hierarchical Bayesian methods for the reconstruction and prediction of random dynamical systems. More recently, there has been increasing research interest in the development of Gaussian process \citep{ki2006gaussian} models for black-box time series modeling \citep{kocijan2005dynamic,turner2009system,bijl2017system, sarkka2018gaussian}.        
     
     A common feature of most of the aforementioned methods is that they are based on the rather strong assumption of Gaussian noise which is rarely true in real world applications and is typically chosen for computational convenience. Other choices, based on scaled mixtures of Gaussians have also been proposed \citep{ishwaran2005spike, polson2010shrink, carvalho2010horseshoe} however, these choices are still parametric capturing only a bounded amount of information in the data. On the other hand, restricting $g$ to belong to a particular class of functions might be difficult when there is not enough information for the problem at hand, leading to models with poor generalization capabilities.
     
    In this paper, we aim to develop hierarchical Bayesian models for the estimation and prediction of random dynamical systems from observed time series data. In particular, we use a nonparametric Bayesian prior namely, the geometric stick breaking process (GSB) prior introduced by \cite{fuentes2010new} which allows for flexible modeling of unknown distributions in order to drop the common assumption of Gaussian distributed noise components. Bayesian nonparametric modeling of the noise process has been successfully attempted before \citep{hatjispyros2009bayesian, merkatas2017bayesian, hatjispyros2019joint} in the context of dynamical reconstruction from observed time series data. Our model is an extension of the model proposed at \cite{merkatas2017bayesian} where we additionally aim to adopt a black-box modeling approach that does minimal assumptions for the data at hand. In particular, we drop the assumption of a known functional form of the deterministic part of the system by leveraging Bayesian neural networks \citep{neal2012bayesian}. We parametrize the data generating mechanism with a feed-forward neural network whose weights and biases are assigned a prior distribution. In the limit of infinite hidden neurons a priori, this neural network converges to a Gaussian process prior \citep[cf.][]{neal2012bayesian} making the approach completely nonparametric. This makes our model flexible and suitable for modeling real time series data, requiring only the order of the underlying Markov process to be known. Recent advances in embedding dimensions and signature learning \citep{fermanian2021embedding} may provide additional information on how to select the lag in the inferential process but this is beyond the scope of this paper.
    
 It should be stressed that a purely Bayesian nonparametric approach to the problem described here would be a density regression approach. In density regression, the functional form of the system and the noise process are assumed unknown and the use of the so called dependent nonparametric priors is required \citep{fuentes2009nonparametric, hatjispyros2018dependent, lijoi2014dependent, mena2011geometric}. However, the construction of such priors can be hard especially for continuous time series modeling because dependencies among the random probability measures changing in time would require analytical solutions to stochastic differential equations or, at least, efficient computational methods. Our framework can be extended to continuous time modeling using more sophisticated neural network architectures, for example, neural ordinary differential equations \citep{chen2018neural}.
    
    The paper is organized as follows. Section~\ref{sec:background} introduces the necessary background on Bayesian nonparametric models and Bayesian neural networks. The proposed model for estimation and prediction of random dynamical systems is derived in Section~\ref{sec:the-model}. Next, in Section~\ref{sec:sampler} we introduce a fast MCMC sampling algorithm for posterior inference based on a Gibbs kernel and we show how the samples taken from the posterior can be used for prediction of future values. Experimental evaluation in simulated and real data is given in Section~\ref{sec:experiments}. Finally, we conclude and discuss the results in Section~\ref{sec:discussion}.

\section{Background} \label{sec:background} 
    In this section we provide the necessary background for the building blocks of the proposed model. We start with the review of Bayesian neural networks \citep{neal2012bayesian} and the Hamiltonian Monte Carlo (HMC) method for posterior inference \citep{kennedy1988hybrid} which is crucial for the update of the neural network weights within our Gibbs sampler. For completeness, an introduction to nonparametric priors \citep{hjort2010bayesian} and their application in mixture modeling is given in Section~\ref{subsec:bnp}. 
    
    \subsection{Bayesian neural networks} \label{subsec:bnn}
        Feed-forward neural networks are popular models for estimating functions from data
        due to their universal approximation property \citep{cybenko1989approximation}. For our purposes we focus on a feed-forward neural network with $L$ layers where each layer consists of $N_l,\, 0\leq l \leq L$ neurons. The parameters of the neural network are given by a set of weight matrices $W^{(l)}:=(W^{(l)}_{ij})^{i=1:N_l}_{j=1:N_l-1}$ and bias vectors $b^{(l)}:=(b^{(l)}_{i})_{i=1:N_l}$ collectively denoted by $\theta,$ that is, $\theta=\{W^{(l)}, b^{(l)}\}_{l=1}^L.$ The network computes functions $y=g(x;\theta)$ by forward-propagating the inputs as follows. In each layer $l=1,\ldots,L$ the network computes the affine transformation $a^{(l)}:=W^{(l)}x^{(l-1)} + b^{(l)}$ to its input and applies an elementwise nonlinearity as $x^{(l)}:=h^{(l)}(a^{(l)})$ where, $x:=x^{(0)}$ and $y:=x^{(L)}.$ 
        Typical choices for the nonlinear function $h$ include the hyperbolic tangent, the logistic function, and more recently, the rectified linear unit function \citep{Goodfellow:2016}. Since we are going to use the network for regression tasks, the function at the output layer $h^{(L)}$ is chosen to be the linear function while for $l<L$ we choose the $\tanh$ activation function.
        
        The neural network becomes \textit{Bayesian} \citep{mackay1992practical, nowlan1992simplifying,neal2012bayesian} when a prior distribution is assigned on the parameters $\theta$ and predictions are made using the posterior predictive distribution of the model. 
        The Bayesian neural network can be written in hierarchical fashion as 
            \begin{align}
                    & y_t \mid x_t, \theta, \lambda \sim \dnormal(y_t\mid g(x_t;\theta), \Lambda^{-1}), \quad 1\leq t \leq n,\\
                    & \theta \sim p(\theta \mid \tau),\,\tau \sim p(\tau),\,\Lambda \sim p(\Lambda), \label{eq:hyper-params}
            \end{align}
            where $(x_t,y_t),\,1\leq t \leq n$ are the training  data, $\Lambda$ is the precision matrix, and $\tau$ is a hyperparameter for the weights. Bayesian inference proceeds by sampling from the full posterior distribution given by
            \begin{align}
               & p(\theta, \Lambda, \tau \mid \{(x_t,y_t)\}_{t=1}^n) \propto p(\Lambda)p(\tau)p(\theta\mid\tau)\prod_{t=1}^n \dnormal(y_t\mid g(x_t; \theta), \Lambda^{-1}).\label{eq:nn_posterior}
            \end{align}
      It is possible to sample the parameters of a Bayesian model via MCMC. As we will see later, in Sect. \ref{subsec:mcmc}, we can sample the parameters of our model by using a combination of MCMC methods namely, Gibbs \citep{geman1984stochastic} and Hamiltonian Monte Carlo (HMC) \citep{kennedy1988hybrid}. In particular, HMC will be used to sample from the full conditional distribution of the neural network parameters $p(\theta\mid \cdots)$ which for now we denote by $p^{|{\cal D}}(\theta).$ 
      After Neal's seminal work \citep{neal2012bayesian}, a sample for the weights is typically obtained from $p^{|{\cal D}}(\theta)$ via HMC \citep{neal2011mcmc} due to its complex geometry. 
      
      In HMC, a new sample is proposed in the parameter space using gradient information instead of simple random walk proposals. Introducing auxiliary variables $\eta$, where $\dim{\eta}=\dim{\theta},$ an ergodic Markov chain with stationary distribution the desired posterior distribution $p(\theta, \eta) = p^{|{\cal D}}(\theta)p(\eta)$ is constructed. This posterior can be seen as the canonical distribution of the joint state $(\theta, \eta)$ corresponding to the Hamiltonian energy function $H(\theta, \eta) = U(\theta) + T(\eta) = -\log p^{|{\cal D}}(\theta) + \frac{1}{2} \eta^\top M^{-1}\,  \eta,$ where $M$ is a positive definite mass matrix.
            
      Given the current state $(\theta, \eta)$ of the sampler, an HMC transition is given by first simulating new variables $\eta'\sim p(\eta)$. Then, the new $\theta'$ is proposed simulating a trajectory according to Hamilton's equation of motion $(\Dot{\theta}, \Dot{\eta})= (\partial_\eta H, -\partial_\theta H)$ using leapfrog integrator, a symplectic integrator which alternates between the following updates
            \begin{align}
                \eta(\kappa + \epsilon/2) &= \eta(\kappa) - \frac{\epsilon}{2} \partial_\theta H(\theta(\kappa), \eta(\kappa))\nonumber\\
                \theta(\kappa + \epsilon) &= \theta(\kappa) + \epsilon\partial_\eta H(\theta(\kappa), \eta(\kappa + \epsilon/2))\nonumber\\
                \eta(\kappa + \epsilon) &= \eta(\kappa + \epsilon/2) - \frac{\epsilon}{2} \partial_\theta H(\theta(\kappa+\epsilon), \eta(\kappa + \epsilon/2)),\nonumber
            \end{align}
      where $\epsilon$ is the discretization step length, $S$ is the number of steps giving a trajectory of length $\epsilon S$, and $\kappa$ is the time variable.

    To account for trajectory simulation errors, a Metropolis acceptance step is included and the proposed state is accepted with probability
    \begin{equation}
    \mathrm{acc}(\theta',\theta)= \min \{1, p(\theta', \eta') / p(\theta, \eta)\}.
    \end{equation}
    
    \subsection{Bayesian nonparametrics} \label{subsec:bnp}
        Bayesian nonparametrics dates back to introduction of Dirichlet processes (DPs) \citep{ferguson1973bayesian} but has been widely used only after the development of Markov chain Monte Carlo (MCMC) \citep{brooks2011handbook} methods. Their main advantage in statistical and machine learning applications is  their ability to adapt their complexity according to the available data, using a potentially infinite number of parameters.  
        In this section we provide some background material on Bayesian nonparametric priors and their application on mixture modeling. 

        A general random discrete distribution on the space $(\Xi, \sigma(\Xi))$  where $\sigma(\Xi)$ is the Borel $\sigma$-algebra of subsets of $\Xi$ can be represented as
        \begin{equation}
            \mathbb{P} = \sum_{k=1}^{K}\pi_k\delta_{\xi_k},\label{eq:discrete-distribution}
        \end{equation}
        such that $\sum_{k=1}^{K}\pi_k=1 \textrm{ a.s.}$ and $(\xi_k)\iid P_0,$ where $P_0$ is a non-atomic $\Xi-$valued measure. Allowing $K=\infty,$ standard results \citep{ongaro2004discrete} can be used to show that random probability measures of the form given in \eqref{eq:discrete-distribution} have full support on the space $\mathcal{P}_{\Xi}$, the space of all probability measures defined on $\Xi.$ These random distributions can be used in the context of mixture models to define nonparametric mixtures. Convoluting the kernel of a parametric family with a random discrete distribution we can define mixture densities as
        \begin{align}
            f(x \mid \PP) = \int_{\Xi} K(x\mid \xi) \PP(\dd \xi),\quad \PP(\dd \xi)= \sum_{k=1}^{K}\pi_k\delta_{\xi_k}(\dd \xi).\label{eq:mixture-generic}\nonumber
        \end{align}
        Depending on the mechanism that the weights $\pi_k$ are constructed and if $K=\infty$ or not, different nonparametric mixture models can be recovered. 
        Specifically, choosing $\PP\sim \dirp(c, P_0)$ to be the Dirichlet process with mass $c>0$ and base measure $P_0$
        \citep{sethuraman1994constructive}, that is,
        \begin{align}
            &\PP= \sum_{k=1}^{\infty} \pi_k \delta_{\xi_k}, \quad \xi_k \iid P_0, \\
            &\pi_k = \phi_j \prod_{j<k} (1 - \phi_j), \quad \phi_1,\phi_2\ldots \iid \dbeta(1,c),\label{eq:stick-breaking-weights}
        \end{align}
        we recover the Dirichlet process mixture (DPM) model \citep{antoniak1974mixtures, escobar1995bayesian} which is commonly used for density estimation problems. 
        
        Random probability measures, with less exotic weights can be obtained if we consider general random discrete distributions. In particular, \cite{ishwaran2001gibbs}, introduced a general class of stick breaking distributions  defined for $\phi_k \sim \dbeta(a_k,b_k)$ as in Equation~\eqref{eq:stick-breaking-weights}, for $a_k,b_k$ such that
         $\sum_{k=1}^{\infty}\log\E\left(1 - \phi_k\right) = -\infty$ \citep{ghosal2017fundamentals}. 
         
        A rather interesting random discrete distribution is the geometric stick breaking process (GSB) \citep{fuentes2010new} which corresponds to a single beta random variable $\phi\sim\dbeta(a,b).$
        \begin{align}
            &\PP= \sum_{k=1}^{\infty} \pi_k \delta_{\xi_k}, \quad \xi_k \iid P_0, \label{eq:gsb-measure} \\
            &\pi_k = \phi(1-\phi)^{k-1}, \quad \phi \sim \dbeta(a,b).\label{eq:gsb-weights}
        \end{align}
        This random probability measure has full support on the space $ \mathcal{P}_{\Xi},$ and in contrast to the Dirichlet process, its weights are almost surely decreasing. Later on, we will model the density of the noise components as a GSB mixture density.
    
        Inference in mixture models proceeds with the introduction of a clustering variable for each observation that indicates the component of the mixture that the specific observation belongs to. However, the state space of the clustering variables is now countably infinite and additional auxiliary variables that make the sum finite have to be introduced (see for example, \cite{walker2007sampling, kalli2011slice} for the DPM model and beyond). We will provide details on this when we will construct the transition density for the observations of our model which we now introduce in Section~\ref{sec:the-model}. Nevertheless, the decreasing nature of the weights allows for faster sampling of the clustering variables leading to faster and less complicated sampling algorithms in contrast to the Dirichlet process counterpart.
    
\section{Autoregressive Bayesian neural networks for time series modeling}\label{sec:the-model}
    In this section we present the Bayesian neural network model with a nonparametric noise component. We start with the development of a nonlinear autoregressive model in Section~\ref{subsec:nlar} and gradually extend it into a fully nonparametric method by introducing Bayesian neural networks as function approximators for the deterministic part and nonparametric noise components based on GSB mixtures.
    
    \subsection{Nonlinear autoregressive Bayesian neural network model}\label{subsec:nlar}
    In this work we consider a discrete time dynamical system 
    \begin{equation}\label{eq:recurrence_eq0}
        y_t = g(y_{t}^{\rho},\theta) + z_t,\quad 1\leq t\leq n,
    \end{equation}
    where $g:\R^{D\times\rho}\times \Theta\to \R^D$ is a nonlinear function of $y_{t}^{\rho}:=(y_{t-1},\ldots,y_{t-\rho})$ parametrized by a vector of control parameters $\theta\in\Theta\subset \R^{m}.$ The random variables $z_t$ represent additive dynamical noise which for now we assume are i.i.d. samples from an unknown density $f(\cdot)$ taking values in $\R^D.$ 

In order to be able to capture nonlinearities in the time series with our model, a sufficiently expressive function should be chosen for $g(\cdot,\theta).$ To this end, we model the function responsible for the generation of the time series $g(\cdot, \theta)$ where $\theta=\{W^{(l)}, b^{(l)}\}_{l=1}^L,$ with a Bayesian neural network of $L$ layers as described in Section~\ref{subsec:bnn}. The prior over the parameters in $\theta$ is taken to be Gaussian parametrized by precision $\boldsymbol{\tau}.$ In particular, we tie each group of weights and biases in a specific layer to their own precision $\tau_W^{(l)},\tau_b^{(l)}$ so that for $1\leq l \leq L:$
        \begin{align}
            W^{(l)} \dotind\dnormal(0,[\tau_{W}^{(l)}]^{-1}), \quad & b^{(l)} \dotind \dnormal(0,[\tau_{b}^{(l)}]^{-1}),
        \end{align}
         where $v\dotind \nu$ denotes independently $\nu-$distributed elements $v_i$ of the vector/matrix $v.$
        The precision parameters $\tau_{\bullet}^{(l)}$ control how much the weights/biases can vary in group $l$ and are further assigned Gamma hyperprior distributions $ \mathrm{Ga}(\tau_{\bullet}^{(l)}\mid\alpha_{\tau_{\bullet}^{(l)}},\beta_{\tau_{\bullet}^{(l)}})$ with mean $\alpha_{\tau_{\bullet}^{(l)}}/\beta_{\tau_{\bullet}^{(l)}}.$ 
        This way a hierarchical model for the weights and biases between layers is defined and information can be shared among different layers, a concept known as ``borrowing of strength'' \citep{gelman2013bayesian,neal2012bayesian}. 

In the next section we adopt a Bayesian nonparametric mixture modeling approach for the density of the noise components. In particular, we aim to model the density of the noise components as an infinite mixture of Gaussian kernels based on the GSB process. This approach is nonparametric and is equivalent with placing a prior over the space of densities. Under the proposed setting, the model presented in \cite{merkatas2017bayesian} is extended to a full nonparametric Bayesian model for system identification. Thus, the only necessary information is the order of the Markov process given in Equation~\eqref{eq:recurrence_eq0} which allows us to form the regressors $y^{\rho}_t.$ 

    \subsection{Nonparametric noise process}
 
        Consider a general discrete random probability measure $\PP  = \sum_{k=1}^{\infty}\pi_k \delta_{\Lambda_k}$ defined as in Equation~\eqref{eq:discrete-distribution}.
        The density of each component $z_t$ of the noise process can be modeled as an infinite mixture of (multivariate) Gaussian kernels
        \begin{align}
            &f(z_t \mid \PP) = \int_{\mathcal{M}^+} \dnormal(z_t\mid \boldsymbol{0}, \Lambda^{-1}) \,\PP(\dd \Lambda)= \sum_{k=1}^{\infty} \pi_k\, \dnormal(z_t\mid \boldsymbol{0}, \Lambda_k^{-1}).\label{eq:noise-mixture}
        \end{align}
        In the equation above, $\mathcal{M}^+$ denotes the set of positive semi-definite $D\times D$ precision matrices defined on $\R,$ and $\boldsymbol{0}$ is the $D$-dimensional zero vector. 
        From the additivity of the errors, the transition density for the observations reads
        \begin{align}
            f(y_t \mid y^{\rho}_t, \theta,(\pi_k)_{k\geq 1}, (\Lambda_k)_{k\geq 1}) = \sum_{k=1}^{\infty} \pi_k \,\dnormal(y_t\mid g(y^{\rho}_t,\theta), \Lambda_k^{-1}). \label{eq:transition}
        \end{align}
        To avoid cluttering up the notation, at this point we adopt a simpler notation for $f,$ namely from now on we denote $f(y_t \mid y^{\rho}_t, \theta,(\pi_k)_{k\geq 1} (\Lambda_k)_{k\geq 1})$ by $f(y_t).$
        Then, based on a sample of size $n$ the likelihood reads
        \begin{align}
            &f(y_1, \ldots, y_{n})
             = \prod_{t=1}^{n} \sum_{k=1}^{\infty} \pi_k \dnormal(y_t\mid g(y^{\rho}_t,\theta), \Lambda_k^{-1}).\label{eq:likelihood}
        \end{align}
      
        Inference in mixture models proceeds with the introduction of a clustering variable $d_t$ for each observation that indicates the component of the mixture that $y_t$ came from. For the prior model, the clustering variables have an infinite state space and each $d_t=k$ is selected with probability $\pi_k,$ that is, $\mathrm{pr}(d_t=k) = \pi_k.$ Any MCMC sampling technique would require an infinite number of steps to sample the $d_t$'s. However, introducing specific type of auxiliary variables, we can make the sum finite \citep{fuentes2010new}. In particular, for each observation, we introduce the random variable $R_t,$ an almost surely finite discrete random variable with probability function $f_R(\,\cdot\mid \phi)$ parametrized by $\phi.$  Conditionally on $R_t$ we force the clustering variables to attain a discrete uniform distribution, that is, $\mathrm{pr}(d_t = k \mid R_t = r) = r^{-1}\indic(\{1,\ldots,r\}).$       
         Augmenting the random density for one observation with $d_t,R_t,$  Equation~\eqref{eq:transition} reads
        \begin{align}
            f(y_t, R_t=r, d_t=k) &= f_R(r \mid \phi) \,f(d_t=k \mid R_t = r)\, f(y_t\mid d_t = k) \label{eq:augmentation}\\
            &=f_R(r \mid \phi)\,\frac{1}{r}\, \indic(r\geq k)\,\dnormal(y_t\mid g(y^{\rho}_t,\theta), \Lambda_k^{-1}). \label{eq:augmented}
        \end{align}
        Marginalizing w.r.t. $(d_t,R_t)$ it is that
        \begin{align}
            f(y_t) & = \sum_{k=1}^\infty \sum_{r=k}^\infty f_{R}(R_t=r\mid \phi)\,r^{-1}\,\dnormal(y_t\mid g(y^{\rho}_t,\theta), \Lambda_k^{-1})\nonumber\\
            & = \sum_{k=1}^\infty \pi_k\, \dnormal(y_t\mid g(y^{\rho}_t,\theta), \Lambda_k^{-1}),
        \end{align}
        where 
        \begin{align}
            \pi_k = \sum_{r=k}^\infty \frac{1}{r}f_R(R_t=r\mid \phi).
        \end{align}
        Depending on the choice of $f_R$ we can construct different random probability measures with decreasing weights. In particular, choosing $f_R(R_t\mid \phi)=\mathrm{Nbin}(R_t=r\mid 2, \phi),\,\phi\in(0,1)$ the weights become geometric
        \begin{align}
            \pi_k &=  \sum_{r=k}^\infty \frac{1}{r} r \phi^2 (1 - \phi)^{r-1}
             =  \phi(1-\phi)^{k-1}.
        \end{align}
        It can be shown \citep{de2020inferential} that negative-binomial distributions with parameter $3$ lead to weights which have the form $\pi_k = \phi(1-\phi)^{k-1}(1+k\phi)/2.$ The choice of Poisson distribution with parameter $\phi>0$, shifted to start at 1, would give  
        \begin{align}
            \pi_k &= \frac{\Gamma(k) - \Gamma(k, \phi)}{\phi\,\Gamma(k)}
        \end{align}
        where $\Gamma(\cdot), \Gamma(\cdot,\cdot)$ are the Gamma and upper incomplete Gamma functions, respectively. 
        
    \subsection{Nonlinear model}
       For our purposes, we choose $f_R(R_t)=\mathrm{Nbin}(R_t=r\mid 2, \phi)$ for the auxiliary variables. Under this choice, and this will make the steps of the Gibbs sampler easier to understand, we can write the proposed nonlinear model in hierarchical fashion as
        \begin{align}
        &    y_{t}\mid y^{\rho}_{t}, \theta, d_t=k, \Lambda \ind \dnormal(y_t\mid g(y^{\rho}_t,\theta), \Lambda_k^{-1}) \nonumber\\
        &W^{(l)} \dotind \dnormal(0, [\tau_{W}^{(l)}]^{-1}),\,\,b^{(l)} \dotind \dnormal(0, [\tau_{b}^{(l)}]^{-1}),\, 1\leq l \leq L\nonumber\\
        &\tau_{W}^{(l)} \ind \mathrm{Ga}(\alpha_{\tau_{W}^{(l)}},\beta_{\tau_{W}^{(l)}}),\,\,\tau_{b}^{(l)} \ind \mathrm{Ga}(\alpha_{\tau_{b}^{(l)}},\beta_{\tau_{b}^{(l)}}),\, 1\leq  l \leq L\nonumber\\
        & d_t = k \mid R_t = r \ind \mathrm{DiscreteUnif}\{1,\ldots,r\}\nonumber\\
        & R_t \iid \mathrm{Nbin}(2,\phi),\quad \phi\sim \dbeta(a_\phi,b_\phi)\nonumber\\
        &\pi_k = \phi(1-\phi)^{k-1}, \quad \Lambda_k \iid P_0,\quad 1\leq k \leq R^*,\nonumber
        \end{align}
        where $R^*:=\max_t \{R_t\}.$
         The hierarchical Bayesian model defined above, is a highly flexible regression model in the sense that a \textit{nonparametric prior} is defined on the space of all possible noise distributions while the functional form of the system is estimated from the \textit{Bayesian neural network}. We will call the model with those properties \textit{NP-BNN} from now on.  In particular, for the NP-BNN the following proposition holds:
        
        \begin{proposition}\label{rem:convergence-to-nonparametric-gp}
            Let $g(x,\theta)$ be a neural network with one hidden layer of size $N_1.$ Assume for the vector $\theta=(W^{(1)},b^{(1)})$ that  $b^{(1)}\iid \dnormal(0, \sigma_{b}^2)$ and the elements of $W^{(1)}\iid \dnormal(0, \sigma_w^2).$ As $N_1\to \infty$ the model converges to a full nonparametric Bayesian regression model.
        \end{proposition}
        
        \begin{proof}
            The prior on the density of the noise components is clearly nonparametric as an infinite mixture of Gaussian kernels convolved with a random probability measure $\PP\sim \Pi$ where $\Pi$ is the geometric stick breaking process with support the space $\mathcal{P}_{\mathcal{M}^+},$ the space of all probability measures defined on $\mathcal{M}^+.$ We note here that this is true for any random probability measure of the class considered in \cite{ongaro2004discrete}.
            
            To show that the regression function becomes nonparametric a priori, the proof is similar to Neal's proof for the convergence of functions computed by Bayesian neural networks to Gaussian processes \citep{neal2012bayesian}.
            Specifically, under the Gaussian prior for the weights  the following holds
            \begin{equation}
                g(y^{\rho}_t,\theta) \sim \mathrm{GP}\left(\boldsymbol{0}, K(y_{t}^{\rho},y_{t'}^{\rho})\right), \quad \text{while } N_1\to \infty,
            \end{equation}
            where $K(x,x')=\sigma_{b}^2+\sigma_w^2\E\left[h_{j}^{(1)}(x)\, h_{j}^{(1)}(x')\right]$ which is the same for $1\leq j \leq N_1.$ Here, $h_j^{(1)}$ is the activation function corresponding to the $j$th neuron in the hidden layer. 
            
            Now, consider for simplicity  scalar observations $y_t$ and a given configuration of the clustering variables $\mathbf{d}:=(d_1,\ldots, d_n).$ Letting $\Sigma_{\mathbf{d}}=\mathrm{diag}\{1/\Lambda_{d_1},\ldots,1/\Lambda_{d_n}\}$ denote the diagonal matrix with elements the variances indicated from the clustering variables, the model is equivalent to a heteroskedastic Gaussian Process regression model \citep{ki2006gaussian}
            \begin{align}
                &  \mathbf{y} \mid \mathbf{y}^{\rho}, \Sigma_{\mathbf{d}}, \theta, \Lambda \ind \dnormal(\mathbf{y}\mid g(\mathbf{y}^{\rho},\theta), \Sigma_{\mathbf{d}}) \nonumber\\
                & g \sim \mathrm{GP}(\boldsymbol{0}, K(y^{\rho}_t, y_{t'}^{\rho
               }))),\nonumber
            \end{align}
            where $\mathbf{y},\mathbf{y}^{\rho}$ are the vectors of outputs (observations) and inputs (lagged versions) respectively. The vector $\theta$ now represents kernel hyperparameters which can be estimated from data either via MCMC or using some gradient based optimization. 
            
            At the infinite limit $N_1\to \infty$, the hierarchical representation of the model reads
            \begin{align}
                &    y_{t}\mid y^{\rho}_t, d_t=k, \theta, \Lambda \ind \dnormal(y_t\mid g(y^{\rho}_t,\theta), \Lambda_k^{-1}) \nonumber\\
                & g(\cdot, \theta) \sim \mathrm{GP}(\boldsymbol{0}, K(y^{\rho}_t, y^{\rho
               }_{t'})),\quad \theta\sim p(\theta)\nonumber\\
                & d_t = k \mid R_t = r \ind \mathrm{DiscreteUnif}\{1,\ldots,r\}\nonumber\\
                & R_t \iid \mathrm{Nbin}(2,\phi),\quad \phi\sim \dbeta(a_\phi,b_\phi)\nonumber\\
        		&\pi_k = \phi(1-\phi)^{k-1}, \quad \Lambda_k \iid P_0,\quad 1\leq k \leq R^*.\nonumber
            \end{align}
        The hierarchical representation above, is the same representation for a Gaussian process autoregressive model with the Gaussian likelihood replaced by a nonparametric likelihood based on the GSB process.
        \end{proof}
        
\section{Posterior inference via MCMC} \label{sec:sampler}

    In this section we introduce a MCMC method for posterior inference with the proposed model. In particular, a Gibbs kernel is provided and all the variables are sampled from their full conditional distributions. However, as is common, we sample the weights via HMC \citep{neal2012bayesian} within the Gibbs procedure. Additionally, using the sampled values for the parameters $\theta$ we extend the model for prediction of future unobserved values. For simplicity, from now on we consider scalar observations and consequently, $\Lambda$ stands for precision instead of precision matrix.

    \subsection{System identification via MCMC}\label{subsec:mcmc}
        A single sweep of the sampling algorithm has to sample the variables $(d_t, R_t)_{t=1}^n,$ the vector of parameters $\theta,$ as well as the weights and atoms of the random measure $\pi_k, \Lambda_k, 1\leq k\leq R^*.$ For density estimation of the noise process, the sampler must additionally swipe over the noise predictive distribution
        \begin{equation}
            \Pi(\dd z_{n+1}\mid z_{1:n}) = \int \Pi(\dd z_{n+1}\mid \PP)\, \Pi(\dd \PP\mid z_{1:n}). \label{eq:noise-predictive}
        \end{equation}
    
        \medskip\noindent{\bf 1.} First, for $k=1,\ldots, R^*$ we construct the geometric weights
        \begin{equation}\label{eq:sample-weights}
            \pi_k = \phi (1 - \phi)^{k-1}.
        \end{equation}
    
        \medskip\noindent{\bf 2.} Then we sample the atoms of the random measure from their posterior distribution. That is, we sample
        \begin{equation}\label{eq:sample-atoms}
            f(\Lambda_k \mid \cdots) \propto p_0(\Lambda_k)\prod_{d_t=k} \dnormal\left(y_t\,\left.\right | \, g(y^{\rho}_t,\theta), \Lambda_k^{-1}\right),
        \end{equation}
        for $1\leq k\leq R^*.$
        Here, $p_0(\cdot)$ is the density of $P_0$ with respect to the Lebesgue measure.
        
        \medskip\noindent{\bf 3.} Having updated the atoms we proceed to the sampling of the clusters. For $t=1,\ldots,n$ the clustering variable is a sample from the discrete distribution
        \begin{align}\label{eq:sample-clusters}
            &\mathrm{pr}(d_t=k\mid R_t = r,\cdots)\propto \pi_k \,\dnormal\left(y_t\mid g(y^{\rho}_t,\theta),\Lambda^{-1}_k\right)\,\indic(k\leq r).
        \end{align}
    
        \medskip\noindent{\bf 4.} Given the clustering variables, the $R_t$'s which make the sum finite have a truncated geometric distribution
        \begin{equation}
            \mathrm{pr}(R_t=r\mid d_t=k) \propto (1-\phi)^{r-1}\indic(k\leq r).\label{eq:sample-slice}
        \end{equation}
    
        \medskip\noindent{\bf 5.} Under the $\dbeta(a_\phi,b_\phi)$ prior for the geometric probability it's full conditional distribution is 
        \begin{equation}
            f(\phi\mid\cdots)\propto \phi^{a_\phi+ 2n-1}(1-\phi)^{b_\phi+\sum_{t=1}^n R_t-n-1}\label{eq:sample-geometric-probability},
        \end{equation}
        a Beta distribution $\dbeta(a_\phi+2n, b_\phi + \sum_{t=1}^n R_t - n).$
    
        \medskip\noindent{\bf 6.} Now our attention goes to the sampling of the parameters $\theta$ of the neural network.  The weights and biases in each layer $l,\,1\leq l \leq L$ are assigned independent Gaussian priors parametrized by precisions $\tau^{(l)}, 1\leq l \leq N_l,$ that is, $f(\theta^{(l)}\mid \tau^{(l)}, \, 1\leq l \leq L) = \dnormal(\theta^{(l)}\mid 0, [\tau^{(l)}]^{-1})$ where $\theta^{(l)}=\{W^{(l)},b^{(l)}\}$ is the collection of parameters on layer $l.$ Thus, the full conditional distribution for the $\theta$ vector is given by
        \begin{equation}
            f(\theta^{(l)}\mid\cdots) \propto \dnormal(\theta^{(l)} \mid 0, [\tau^{(l)}]^{-1})\prod_{t=1}^n \dnormal\left(y_t\,\big|\,g(y^{\rho}_t,\theta), \Lambda_{d_t}^{-1}\right).\label{eq:sample-theta}
        \end{equation}
        This is a nonstandard distribution because the nonlinear function on $\theta$ appearing in the product of Gaussians. We sample this distribution via a Hamiltonian Monte Carlo step within the Gibbs kernel as described in Section~\ref{subsec:bnn}.
    
        \medskip\noindent{\bf 7.} Under the Gamma prior for the hyperparameters $\tau_\bullet^{(l)},$ we can update the precisions tied to the prior for each layer by sampling from the semi-conjugate model
        \begin{equation}
            f(\tau_\bullet^{(l)}\mid \cdots) \propto \mathrm{Ga}(\tau_{\bullet}^{(l)}\mid \alpha_{\tau_{\bullet}^{(l)}},\beta_{\tau_{\bullet}^{(l)}})\, \dnormal(\theta^{(l)}\mid 0, [\tau_\bullet^{(l)}]^{-1}),\quad 1\leq l \leq L,
        \end{equation}
        which is a Gamma distribution easy to sample from. 
    
        \medskip\noindent{\bf 8.} Finally, we take a sample from the noise predictive given in \eqref{eq:noise-predictive}. At each iteration we have weights and locations $(\pi_k,\Lambda_k)$ and we sample a uniform random variable $u\sim \mathrm{U}(0,1).$ We take $\Lambda_k$ for which
        \begin{align}
            \sum_{\ell=1}^{k-1} \pi_\ell < u \leq \sum_{\ell=1}^{k}\pi_\ell.\label{eq:sample-predictive}
        \end{align}
        If we run out of weights we sample $\Lambda_k$ from the prior $P_0(\Lambda_k).$ Eventually, a sample from the noise predictive is taken from $N(0, \Lambda_k^{-1}).$
        
    \subsection{Prediction of ${\cal J}$ future values}\label{sec:prediction}
        Under the assumption that the order of the Markovian process responsible for the generation of data is known, it is possible to make predictions for the future unobserved values. In particular, predicting the values in a finite prediction horizon ${\cal J}>0,$ requires to sample from the conditional distribution of each future unobserved value. Due to the Markovian structure imposed in \eqref{eq:recurrence_eq0}, the conditional distribution for future unobserved $1\leq j \leq {\cal J}$ values is given by
        \begin{align}
            &f(y_{n+j}, R_{n+j}=r, d_{n+j}=k) = f_R(r \mid p)\,\frac{1}{r}\, \indic(r\geq k)\, \dnormal(y_{n+j}\mid g(y^{\rho}_{n+j},\theta), \Lambda_k^{-1}) \label{eq:predictions-density}
        \end{align}
        where now, the future unobserved values $y_{n+1},\ldots, y_{n+{\cal J}}$ have associated auxiliary variables $R_{n+j}, d_{n+j}$ for $1\leq j \leq {\cal J}.$ This however, results in distributions with very large variance making prediction essentially infeasible. The reason for this is that under the nonparametric noise assumption, future data will with high probability enable new clusters that don't include any observations from the training data. Then the precision $\Lambda_k$ in \eqref{eq:predictions-density} is a sample from the prior which in a non-informative setting is a Gamma distribution with typically large variance. 
    
        It is possible though to use the samples obtained from the MCMC described in Section~\ref{subsec:mcmc} in order to compute functions $\hat{g}$ in a Monte-Carlo fashion. In particular, at each step of the MCMC we have samples $\theta^{(i)}$ for $i=1,\ldots, \mathcal{I}$ where $\mathcal{I}$ is the total number of iterations. We can plug in these samples in the neural network and make predictions for the future ${\cal J}$ values as
        \begin{align}
            \hat{y}_{n+j}^{(i)} &= g(y^{\rho}_{n+j};\theta^{(i)}),\quad j=1,\ldots,{\cal J},\label{eq:first-pred}
        \end{align}
         where
        \begin{align}
           y^{\rho}_{n+j}:=(\hat{y}_{n+j-1}^{(i)},\ldots,\hat{y}_{n+1}^{(i)},y_n,\ldots, y_{n-\rho+1}).\nonumber
        \end{align}
        A biased estimate can be obtained by averaging over the sample of computed functions as
        \begin{equation}
            \hat{y}_{n+j} = \frac{1}{\mathcal{I}} \sum_{i=1}^{\mathcal{I}} g(y^{\rho}_{n+j};\theta^{(i)}),\quad j=1,\ldots,{\cal J}.
        \end{equation}
        The associated Monte-Carlo error can be computed from the sample variance of predicted values. It is worth noting here that although these predictions are biased, they produce better results with less variance than the predictions using Equation~\eqref{eq:predictions-density}. Similar predictors have been used before in \cite{liang2005bayesian}.
        
\section{Experiments} \label{sec:experiments}

 In this section we test the performance of the proposed methodology in simulated and real time series under the assumption of a known order for the underlying Markov process. For the simulated example we use dynamically noise-corrupted time series generated from the logistic map. As real data, we use the Canadian Lynx data, quite popular in time series literature. 
 
\subsection{Experimental setting}

In all the experiments below we adopt the same network architecture and choose same prior specifications. The neural network has the structure $(N_0, 10, 1)$ with $\tanh$ activation functions in the hidden layer, where the number of inputs $N_0$ will be specified by the lag in each example. As a prior for the geometric probability of the GSB process we assign a uniform $\mathrm{Be}(\phi\mid 1,1)$ distribution. The precisions $\tau^{(l)}_\bullet$ for each layer's $1\leq l \leq L$ weights and biases
    are having a $\mathrm{Ga}(\tau^{(l)}_\bullet\mid 5, 5)$ Gamma hyper-prior assigned. These choices for the hyperparameters will moderate the precisions of the Normal priors used for the weights of the network \citep{liang2005bayesian}. The base measure of the GSB process is a Gamma measure $\mathrm{Ga}(\dd \Lambda\mid a_\Lambda,b_\Lambda)\propto \Lambda^{a_\Lambda-1}\exp{(-b_\Lambda \Lambda)}\dd \Lambda$ whose parameters will be specified in each example separately along with HMC tuning parameters $(\epsilon,S).$  All MCMC samplers ran for $40,000$ iterations where the first $2,000$ samples are discarded as burn-in period and a sample is kept (thinning) after 50 iterations.

    Depending on each example's nature, we perform a comparison with related models in terms of prediction accuracy. Letting $\hat{y}$ be the predicted values, we compare the models under the following metrics:
    \begin{enumerate}
        \item Mean Square Error (MSE): $\sum_{t=n+1}^{n+{\cal J}} (\hat{y}_t-y_t)^2 / {\cal J}.$
        \item Root Mean Square Error, the square root of MSE.
        \item Mean Absolute Error (MAE): $\sum_{t=n+1}^{n+{\cal J}} |\hat{y}_t-y_t| / {\cal J}.$
        \item Mean Absolute Prediction Error (MAPE) defined as:
            $\frac{1}{{\cal J}}\sum_{t=n+1}^{n+{\cal J}} \left|\frac{\hat{y}_t-y_t}{y_t}\right|$
        \item Theil's $U$ statistic \citep{pindyck1998economentric}: $U = \frac{\left(\frac{1}{{\cal J}}\sum_{t=n+1}^{n+{\cal J}} (\hat{y}_t-y_t)^2 \right)^{1/2}}{\left(\frac{1}{{\cal J}}\sum_{t=n+1}^{n+{\cal J}} \hat{y}_t^2 \right)^{1/2}+ \left(\frac{1}{{\cal J}}\sum_{t=n+1}^{n+{\cal J}} y_t^2 \right)^{1/2}}.$
        Theil's statistic $U$ takes values in $[0,1]$ predictions such that $U=0$ means perfect prediction.
    \end{enumerate}
    \subsection{Noisy logistic map}
        In this section, we consider dynamically noise corrupted time series generated from the noisy logistic map whose deterministic part is given by $T(x;\mu) = 1 - \mu x^2$
        for the value of $\mu=1.71,$ for which the system exhibits chaotic behavior. The dynamical noise process is described from the 2-component mixture of normal kernels given by
        \begin{equation}
                f(z) = \frac{1}{3}\dnormal(z\mid 0, \sigma_1^2) +     \frac{2}{3}\dnormal(z\mid 0, \sigma_2^2) \label{eq:logistic-noise} 
        \end{equation}
        with $\sigma_1 = 0.04,\,\sigma_2 = 10^{-4}.$ We have generated $n=210$ observations from the random recurrence
\begin{equation}
       x_t = T(x_{t-1}; \mu) + z_t, \quad z_t\iid f(z), \quad 1 \leq t \leq n,\label{eq:random-logistic}
\end{equation}
where the first $200$ observations are used for training and the last $10$ observations have been held out to test the prediction capability of our model. 
        
Because the observations have been generated from the recurrence given in  Equation~\eqref{eq:random-logistic}, the input layer consists of $N_0=1$ units. The parameters of the gamma base measure of the GSB process is chosen according to \citet{merkatas2017bayesian} to be a $\mathrm{Ga}(a_\Lambda, b_\Lambda)$ with $a_\Lambda=3,\, b_\Lambda=0.001.$ Under those values the model becomes data hunting and enables a large number of clusters which leads to a more detailed description of the true noise process. For the HMC tuning parameters we set $(\epsilon, S) = (0.0015, 3).$
        
For comparison we use the geometric stick breaking reconstruction model (GSBR) \citep{merkatas2017bayesian} and an auto-regressive Bayesian neural network (AR-BNN) with Gaussian noise \citep{nakada2005bayesian}. These two methods are specifically tailored for reconstruction and prediction of nonlinear and chaotic dynamical systems and essentially are the building blocks of our model. In particular, GSBR model assumes a polynomial functional form $g(x;
        \theta) = \sum_{p=0}^{P} \theta_p x^{p}$ where $\theta:=(\theta_0,\ldots,\theta_P)$ are the coefficients of the modeling polynomial and the density of the noise components is modeled as a GSB mixture. The AR-BNN is essentially a special case of our model when, roughly speaking, the infinite mixture collapses to a singular mixture of 1-Gaussian distribution.
        
        The performance of all models in terms of the evaluation metrics defined before is given in Table~
        \ref{tab:logistic-compare}. Prediction of chaotic time series is a difficult task itself however, we can see that NP-BNN outperforms all the other models, a more or less anticipated result. On one hand, GSBR does predictions in a similar manner as in Equation~\eqref{eq:predictions-density} which comes with the drawbacks discussed in Section~\ref{sec:prediction}. On the other hand, the AR-BNN would perform well in principle however, the modeling of the noise component with a Gaussian distribution leads to overestimated variance (see Figure~\ref{fig:logistic-noise}) which breaks the quality of the predictions.  
        \begin{table}[H]
        \begin{center}
        \caption{Comparison of the NP-BNN, GSBR and AR-BNN models in terms of prediction accuracy on 10 out-of-sample future values of the chaotic logistic time series data. NP-BNN outperforms the other models under all metrics.}\label{tab:logistic-compare}
        \begin{tabular}{l l l l }
            \noalign{\smallskip}\hline
            Method & NP-BNN & GSBR & AR-BNN  \\
            \hline
            MSE & $0.1802$ & $0.4324$ & $0.2719$ \\
            RMSE & $0.4245$ & $0.6575$ &$0.5214$ \\
            MAE & $0.3503$ & $0.5630$ & $0.4853$ \\
            MAPE (\%) & $102.3866$ & $110.6113$ & $123.5276$ \\
            Theil's $U$ & $0.3578$&  $0.6150$ & $0.5161$  \\
        \noalign{\smallskip}\hline
        \end{tabular}
        \end{center}
        \end{table}
        \begin{figure}[H]
            \begin{center}
            \includegraphics[width=0.6\textwidth]{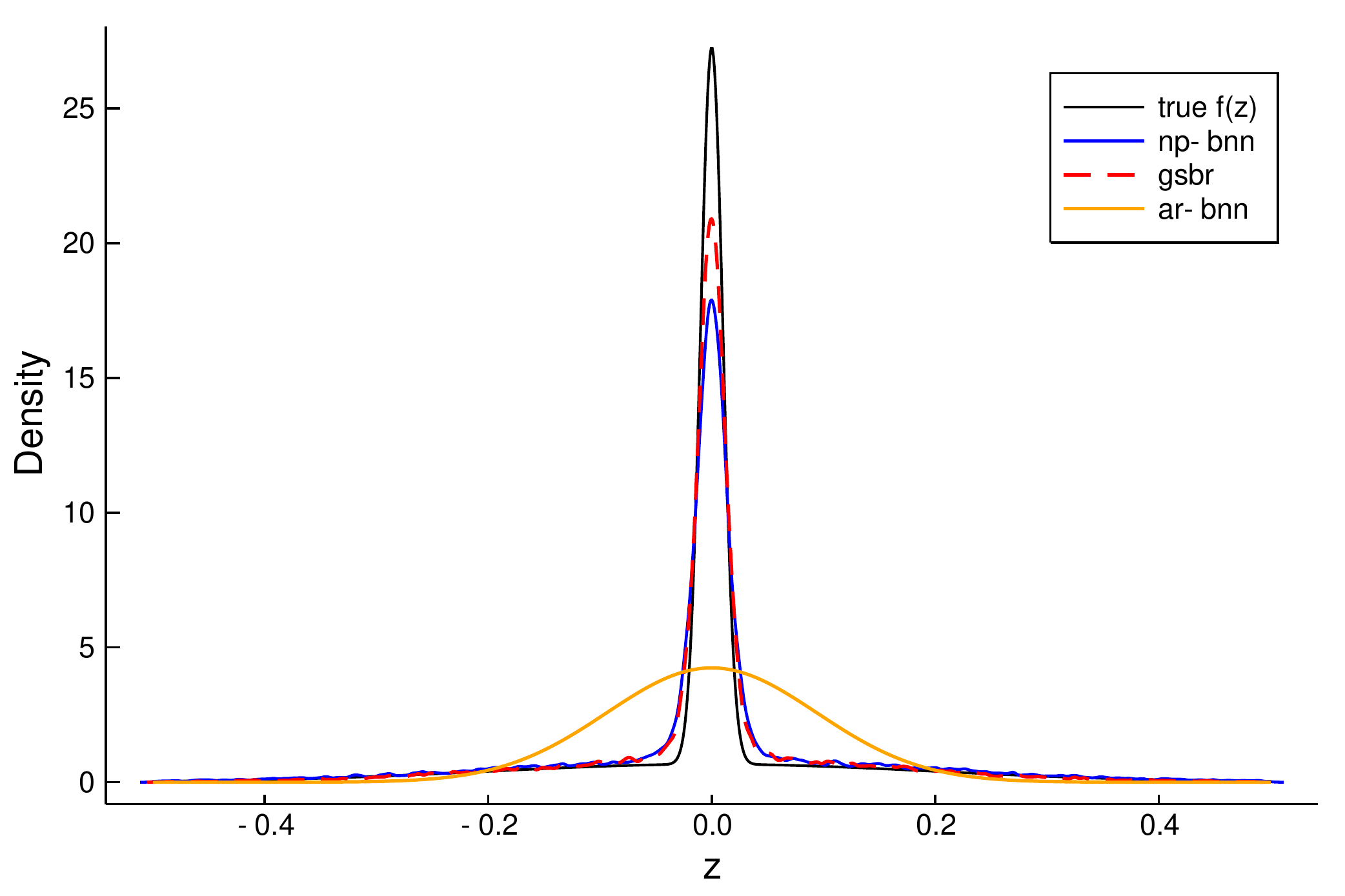}
            \end{center}
            \caption{Kernel density estimator for the posterior predictive density based on the MCMC samples taken from NP-BNN (blue solid line), GSBR (red dashed line) and AR-BNN (orange solid line). The true density of the noise give in Equation~\eqref{eq:logistic-noise} is superimposed with black color.}
            \label{fig:logistic-noise}
        \end{figure}
        
        In Figure~\ref{fig:logistic-noise} we present the density of the noise component, that is, the 2-normal mixture given in Equation~\eqref{eq:logistic-noise} (black solid line). Together we superimpose kernel density estimators (KDE) based on the posterior noise-predictive samples obtained from the NP-BNN (blue solid curve) and the GSBR model \citep{merkatas2017bayesian} (red solid line). The density of a Gaussian distribution with precision estimated from the AR-BNN \citep{nakada2005bayesian} with Gaussian noise is plotted in orange. For this heavy tailed distribution, as expected, models with nonparametric modeling of the noise components provide estimators nearly indistinguishable and very close to the true density. In contrast, the AR-BNN model overestimates the variance in order to capture the heavy tails and fails to estimate the intense peak of the mixture.

        \begin{figure}[H]
            \centering
            \begin{subfigure}[H]{0.45\textwidth}
                \centering \includegraphics[width=\textwidth]{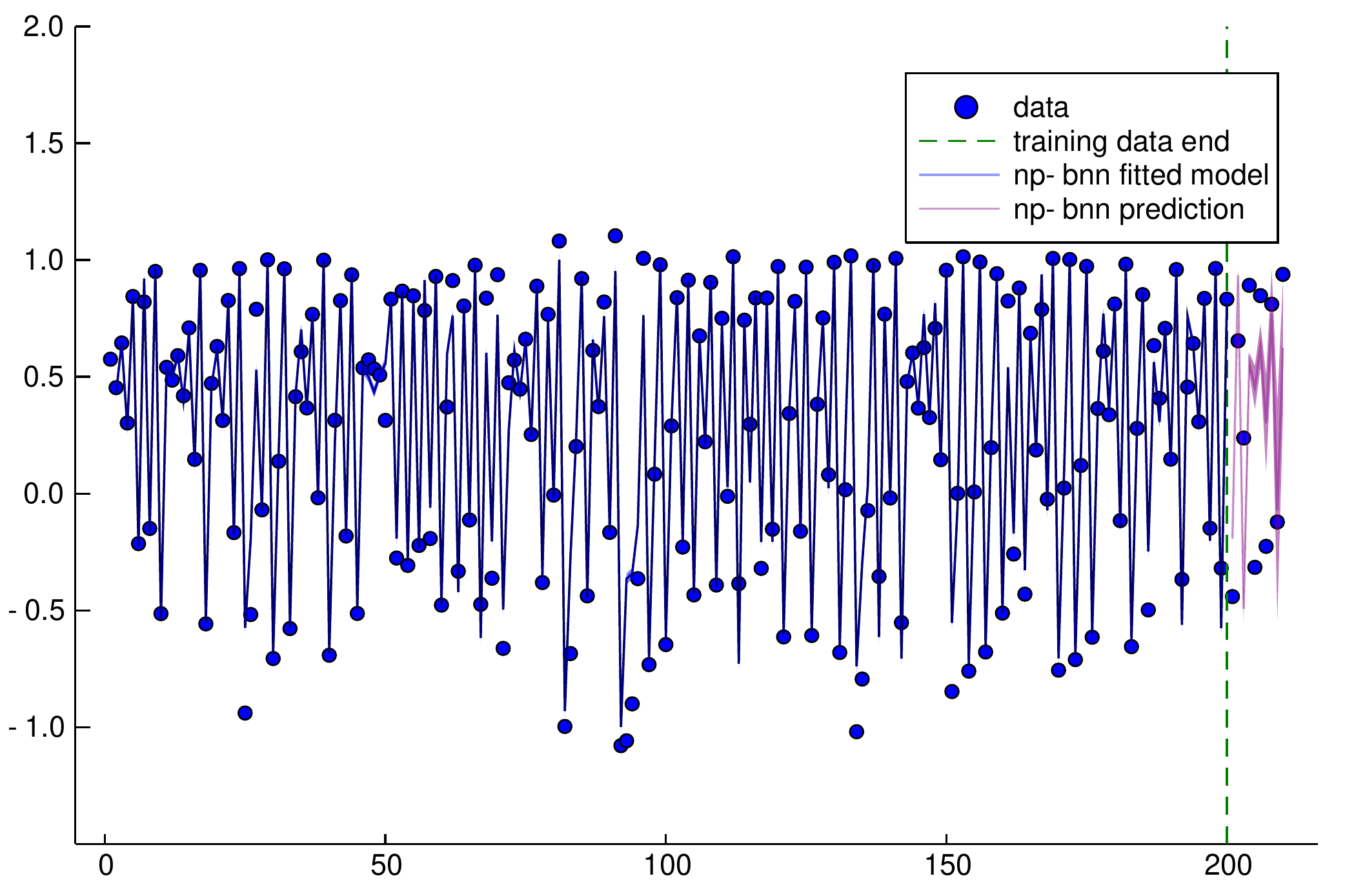}
                \caption{}
                \label{fig:logistic-identification}
            \end{subfigure}
            \begin{subfigure}[H]{0.45\textwidth}
                \centering
                \includegraphics[width=\textwidth]{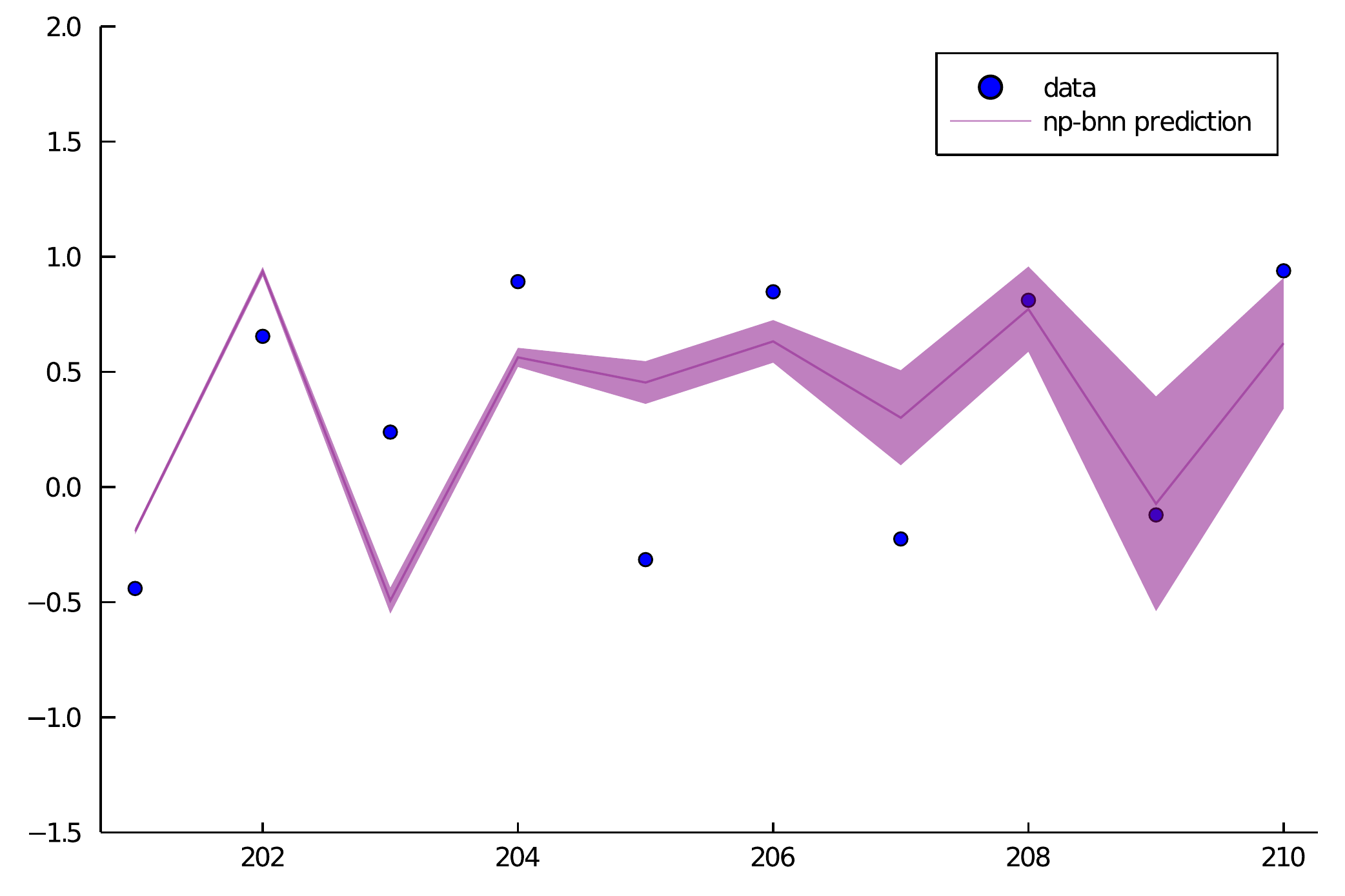}
                \caption{}
                \label{fig:logistic-zoomid}
            \end{subfigure}
            \caption{True time series generated from the noisy logistic map (blue dots) and the system identified from the proposed method (blue line) are shown in panel \ref{fig:logistic-identification}. The shaded lines represent sampling variability due to MCMC sampling. Additionally we present the predicted values for the future 10 values in purple. Panel \ref{fig:logistic-zoomid} is a magnification of the prediction in time steps 201 to 210.} 
            \label{fig:two-id}
        \end{figure}
        
    \subsection{Real data: Canadian Lynx}
        As a real data example we take the Canadian Lynx data which consists of the annual trappings of lynx in the Mackenzie River District of North–West Canada for the period from 1821 to 1934 resulting to a total of 114 observations. We aim to identify the process responsible for the generation of the data using the first 100 observations and, additionally, test our model's predictive capability on the last 14 observations.
        
        As is common, we take the $\log_{10}$ transformation of the dataset. This data set has been widely used as a benchmark for nonlinear time series estimation and prediction before \citep{chan2001chaos,liang2005bayesian} and can be easily accessed from the web. The choice of lag is a model selection problem typically addressed using some information criterion like the Akaike information criterion (AIC) \citep{akaike1974new}. For an AR model, AIC chooses the best lag to be $\rho=11$ but, as in \cite{liang2005bayesian}, we believe that this might lead to overfitting and we assume lag $\rho=2.$ 
        
        For this example, we set the number of inputs in the neural net $N_0=2$. This corresponds to a autoregressive model with $\rho=2.$ Similar approach was taken also in \cite{liang2005bayesian}. The base measure of the GSB process is chosen to be a gamma $\mathrm{Ga}(a_\Lambda, b_\Lambda)$ with $a= b=0.05.$ Here, we set $(\epsilon, S)=(0.005, 20).$
        
        The GSBR model is inappropriate for modeling real data, at least without first extending it to include the coefficients of a full quadratic model \citep{chan2001chaos} so it is excluded from the comparison for fairness. We add in the comparison an AR model for which the best lag is $\rho=11$ according to AIC, and a feed forward neural network trained with backpropagation using the same number of hidden neurons and lagged values in the input layer. That is we choose 10 neurons and $\rho=2$ for straightforward comparisons.  
        
        In Figure~\ref{fig:lynx-clusters} we present the ergodic means for the number of clusters estimated during the fitting process from which we conclude that the algorithm is mixing well and converges to the stationary distribution.
        
        \begin{figure}[H]
            \begin{center}               
            \includegraphics[width=0.6\textwidth]{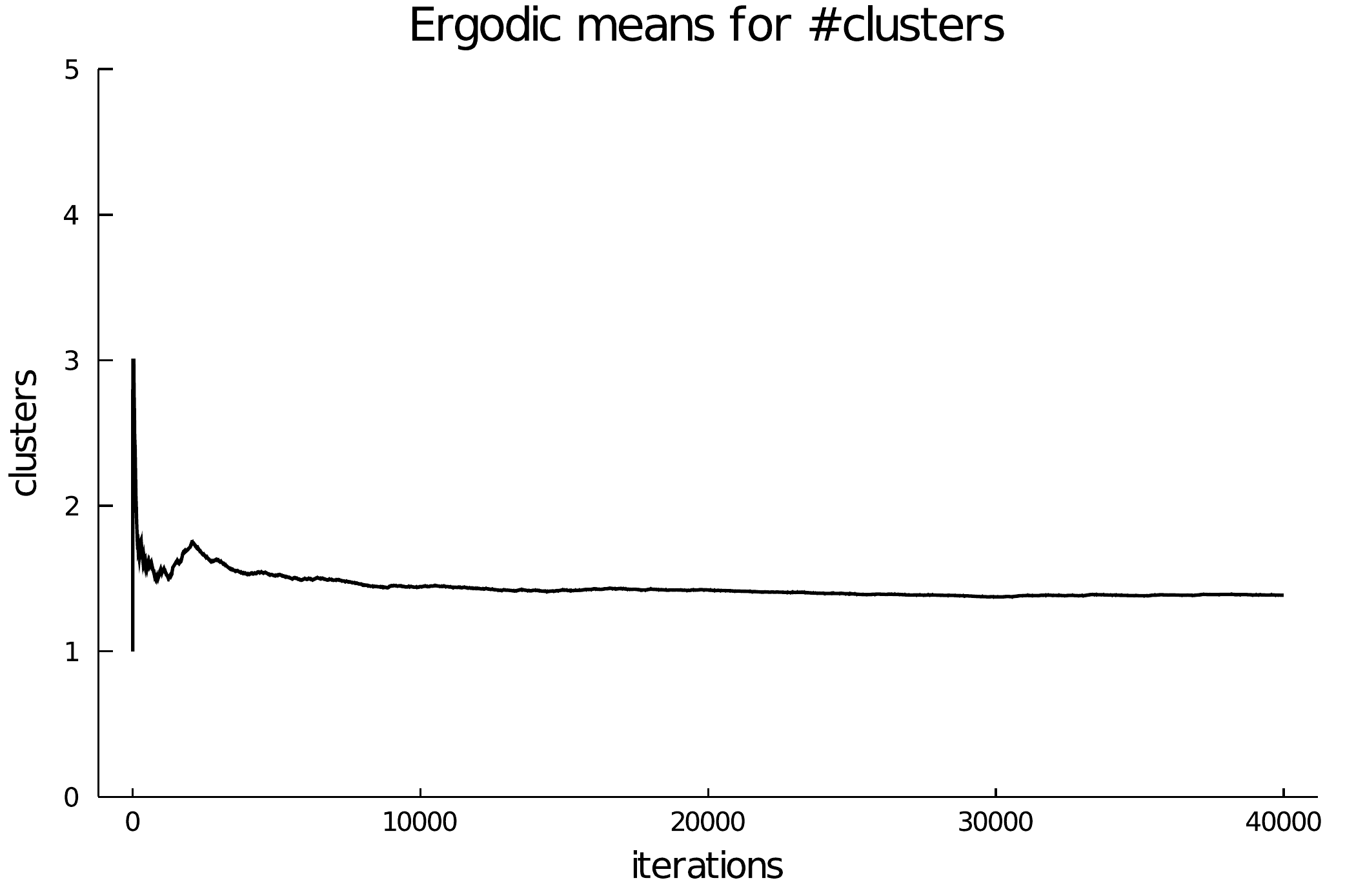}
            \end{center}
            \caption{Ergodic means for the number of the active clusters during fitting.}
            \label{fig:lynx-clusters}
        \end{figure}
        
        In Figure~\ref{fig:lynx-fit-predictions} we represent the model fitted from the proposed method NP-BNN and the out of sample predictions for the next 14 values. The NP-BNN model outperforms all models in this comparison study with additional advantage that it can be used without any hypotheses about linearity of the time series or Gaussianity of its distribution. A comparison in terms of the evaluation metrics introduced in the beginning of this section is given in Table~\ref{tab:lynx-compare}. 
    
        \begin{figure}[H]
            \begin{center}
            \includegraphics[width=0.65\textwidth]{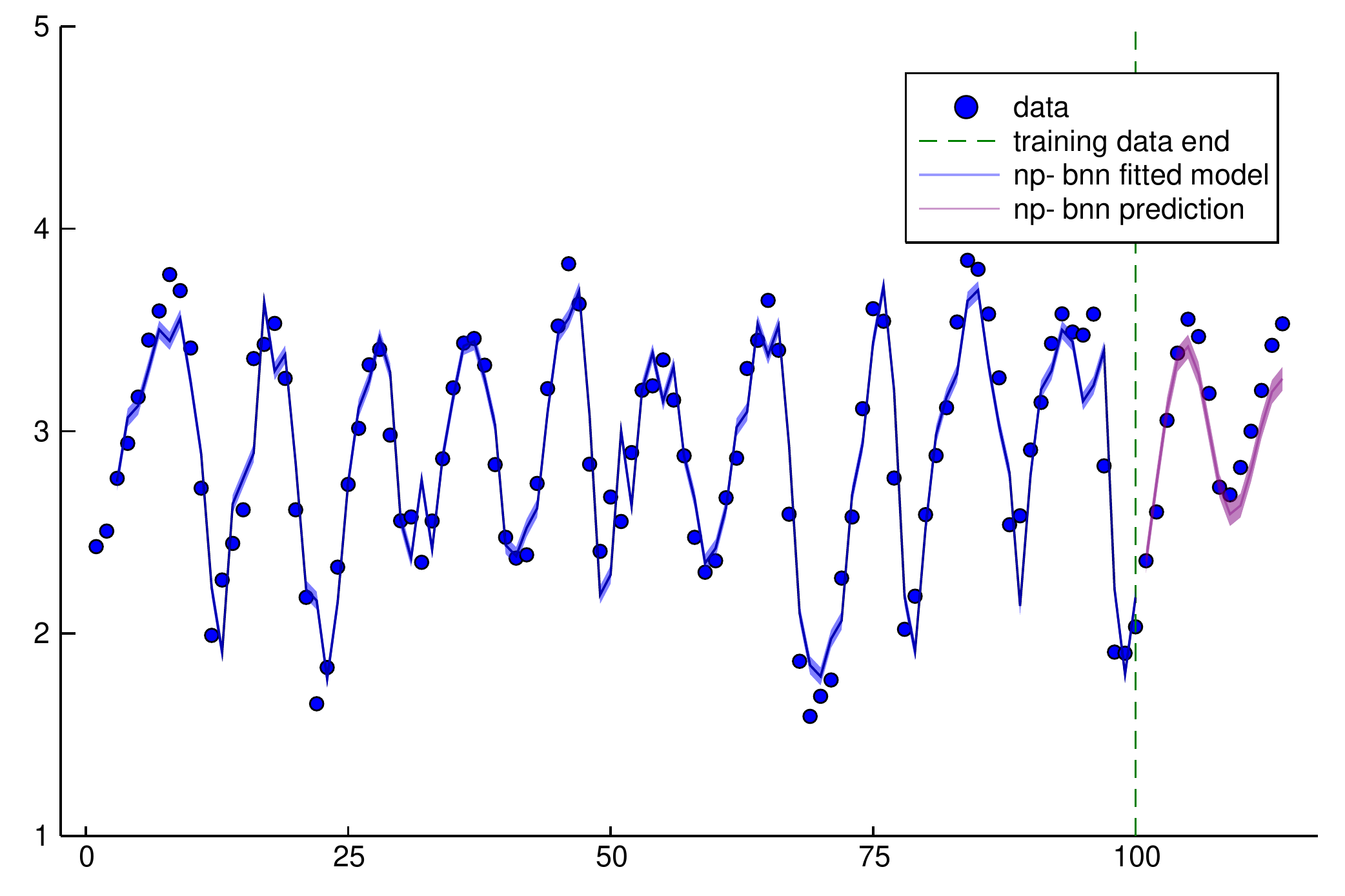}
            \end{center}
            \caption{Fitted time series and 14-step ahead predictions for the Canadian Lynx data (blue dots). The blue line is the mean posterior function computed by the Bayesian neural network and the shaded blue lines represent the uncertainty associated with the function values based on sampled functions. In purple we represent the mean posterior prediction function for the 14 future observations with their associated uncertainties (purple-shaded lines).}
            \label{fig:lynx-fit-predictions}
        \end{figure}
        
        \begin{table}[H]
        \begin{center}
        \caption{Evaluation metrics on Canadian Lynx data based on the prediction for the 14 future values. NP-BNN outperforms the models in the comparisons. For each model, the number in parentheses denotes the lag used for estimation.}\label{tab:lynx-compare}
        \addtolength{\tabcolsep}{-3pt}
        \begin{tabular}{l l l l l l }
            \noalign{\smallskip}\hline
            Method & NP-BNN(2) & AR-BNN(2) & ANN(2) & AR(11) \\
            \hline
            MSE & 0.0249 & 0.0897 & 0.0371 & 0.0822  \\
            RMSE & 0.1579 & 0.2994 & 0.1926 & 0.2866 \\
            MAE & 0.1350 & 0.2273 & 0.1487 &  0.2374  \\
            MAPE (\%) & 4.2727 & 6.9048 & 4.920 & 7.995 \\
            Theil's $U$ & 0.0260 & 0.0502 & 0.0318 & 0.0476 \\
        \noalign{\smallskip}\hline
        \end{tabular}
        \end{center}
        \end{table}

\section{Discussion} \label{sec:discussion}
    We have proposed an auto-regressive Bayesian neural network which assumes nonparametric noise process (NP-BNN) for the estimation and prediction of nonlinear time series models based on the assumption that the order of the underlying Markovian process is known.
    An MCMC sampler for posterior inference  applicable in data sets with small sample size for which the noise process might depart from normality has been devised.
    In the limit of infinite neurons in the hidden layers, we have shown that this model converges to fully nonparametric Bayesian model enjoying all the flexibility of Gaussian process regression models without however, the need of matrix inversions, a characteristic that makes scaling of GPs difficult. Numerical results on simulated and real data sets indicate that the proposed model gives results comparable to the more traditional methods used for nonlinear time series modeling.
    
    During modeling of dynamical phenomena, it is of special interest to estimate the initial state of the system additionally to the system parameters. It would be interesting to modify the proposed methodology so that a full conditional distribution for the initial unknown states when the process responsible for the generation of the time series is non-reversible. This however would require to invert the system $g$ which is not always feasible when $g$ is a neural network. To see this, consider for simplicity a prior $p(y_0)=\indic(-\zeta \leq y_0\leq \zeta) $ for the initial condition $y_0$ and for simplicity, a Gaussian noise model. Then under the proposed framework, the full conditional for $y_0$ reads
    \begin{equation*}
        p(y_0\mid\cdots) \propto p(y_0)\,\dnormal(y_1 \mid g(x_0,\theta), \Lambda^{-1}).
    \end{equation*}
    To devise an exact algorithm based on Gibbs sampling, auxiliary variables can be introduced \citep{damlen1999gibbs} such that 
    \begin{align}
        p(y_0,v_0\mid\cdots) &\propto \indic(-\zeta \leq y_0\leq \zeta) \,\exp\left\{-\frac{\Lambda}{2} v_0^2\right\}  \,\indic(v_0\geq \left(y_1 -g(y_0,\theta))^2\right).\nonumber
    \end{align}
    The full conditional for the auxiliary $v_0$ is clearly a truncated exponential distribution. For $y_0,$ one has to resort on solving the inequalities 
    \begin{align}
        &p(y_0\mid v_0 \cdots)\propto\indic(-\zeta \leq y_0\leq \zeta)\,\indic(y_1-\sqrt{v_0}\leq g(y_0,\theta)\leq y_1+\sqrt{v_0}).\nonumber
    \end{align}
    Letting 
    \begin{equation*}
        Z^{-} = \min\left\{-\zeta, g^{-1}\left(y_1-\sqrt{v_0},\theta\right)\right\},
    \end{equation*}
    and
    \begin{equation*}
    Z^+=\max\left\{\zeta, g^{-1}\left(y_1+\sqrt{v_0},\theta\right)\right\},
    \end{equation*}
    the full conditional distribution for $x_0$ is given by
    \begin{align}
        p(y_0\mid v_0 \cdots) = \mathrm{U}\left(y_0\mid Z^{-},Z^{+}\right),\nonumber
    \end{align}
    where $U$ denotes the uniform density.
    In this case a modeling approach would be to utilize normalizing flows as a model for $g$ which form a class of invertible neural networks \citep{rezende2015variational}.

\section*{Acknowledgements}
The authors would like to thank Academy of Finland (project 321891) for funding the research.

\bibliographystyle{apa}
\bibliography{refs}
\end{document}